\documentclass[12 pt]{amsart}
\usepackage{harpoon}
\usepackage{threeparttable}
\usepackage{tabularx,multirow,bigdelim}
\usepackage{algorithm,algorithmic}
\usepackage{graphics,graphicx,subfigure}
\usepackage{graphicx}
\usepackage{subfigure}\usepackage [latin1]{inputenc}
 \newtheorem{theorem}{Theorem}
 \newtheorem{corollary}{Corollary}
 \numberwithin{equation}{section}

\begin{document}
\title[Verifying full quantum network nonlocality in arbitrary configurations by nonlinear Bell-like inequalities]{Verifying full quantum network nonlocality in arbitrary configurations by nonlinear Bell-like inequalities}

\author{Shuyuan Yang, Jinchuan Hou, Kan He$^*$, Mingxing Luo$^*$}
\address[Shuyuan Yang]{School of Mathematics, North University of China, Taiyuan {030051}, China\\
   School of Cyberspace Science and
Technology, Beijing Institute of Technology, Beijing {100081}, China}

 \email{yangshuyuan2000@163.com}

\address[Jinchuan Hou]{College of Mathematics, Taiyuan University of Technology, Taiyuan,
030024, P.R. China}

 \email{jinchuanhou@aliyun.com}

\address[Kan He$^*$]{College of Mathematics, Taiyuan University of Technology, Taiyuan,
030024, P.R. China}

 \email{hekanquantum@163.com}

\address[Mingxing Luo]
{Southwest Jiaotong University, School of Information Science and Technology, Chengdu {610031}, China\\
Hefei National Laboratory, University of Science and Technology of China, Hefei {230088}, China}
\email{ mxluo@swjtu.edu.cn }

\thanks{ $^{*}$Corresponding author}

\begin{abstract}
Full quantum network nonlocality (FQNN) describes a scenario where all sources in a network are nonlocal. Existing criteria of FQNN can only be verified in star networks by violating a single Bell-like inequality. Here we propose a method that certifies FQNN in general quantum networks using only a single Bell-like inequality. We show that the topological obstacle to one-shot detection can be overcome by expanding the original network with a carefully chosen number of auxiliary local sources and parties. The correlations of the enlarged network are then tested with a single inequality; a violation implies that all original sources must be nonlocal. Our approach provides an efficient, experimentally friendly way to verify FQNN in any network topology.

\end{abstract}
\maketitle

\section{Introduction}

Quantum networks distribute entangled states among spatially separated parties, enabling a wealth of quantum communication and computation tasks~\cite{net1,net2,net3,net4,net5,net6,net7,net8,net9,net10,Tavakoli_2022, hou1, hou2}. Characterizing the nonlocal correlations that such networks can produce has therefore become a central problem. Numerous Bell-type inequalities have been developed for specific topologies, including entanglement-swapping~\cite{PhysRevLett.104.170401,PhysRevA.85.032119}, chain~\cite{Mukherjee2015,PhysRevA.102.052222}, star~\cite{PhysRevA.90.062109}, polygon~\cite{PhysRevLett.123.140401,Jing2019}, tree-shaped~\cite{PhysRevA.104.042405}, arbitrary noncyclic~\cite{PhysRevLett.116.010402,PhysRevA.93.030101}, and even general networks~\cite{PhysRevLett.120.140402}.
Violating any such inequality guarantees that at least one source in the network is nonlocal, but it does not certify that every source is nonlocal. The stronger property--full network nonlocality (FNN)--was introduced by Pozas-Kerstjens \textit{et al.}~\cite{PhysRevLett.128.010403} to describe correlations that can only arise when all sources are nonlocal,  ruling out any correlations generated by the hybrid classical and no-signaling models. For instance, the triangular network of Fig.~\ref{1}(a) contains one local source and therefore cannot produce FNN.

\begin{figure}
\centering
{\includegraphics[width=4.3in]{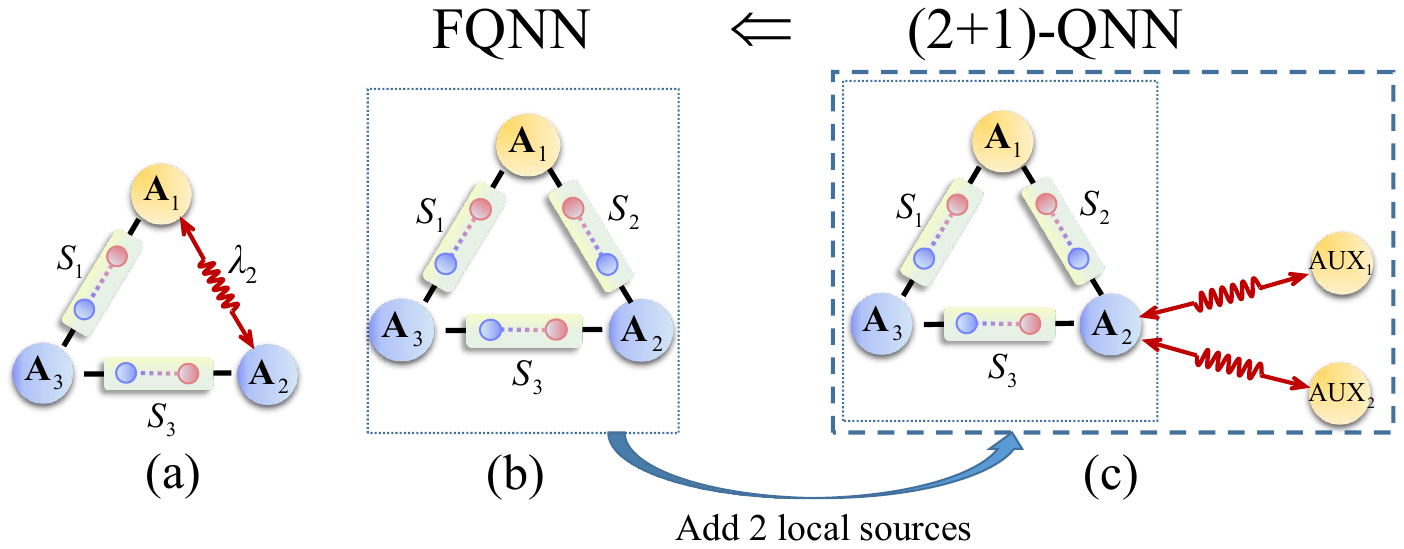}}
  \caption{ The triangular quantum network consists of three sources and three parties:  $S_1,S_2,S_3$ and  ${\bf A}_1,
{\bf A}_2,{\bf A}_{3}$. The yellow circle represents the independent party. The source indicated by the red arrow is a local state, while $(\circ\cdots\circ)$ denotes a nonlocal state. (a) A triangular network comprising two nonlocal sources $S_1$, $S_2$ and one local source, where the local source is described by a hidden variable $\lambda_2$.
 (b)  A triangular network consisting of three nonlocal sources $S_1$, $S_2$, $S_3$.  To witness its FQNN, we add $2$ local sources to ${\bf A}_2$, which results in the new network shown in (c). Now, if network (c) produces $3$-QNN, this implies that network (c) contains at most $2
  $ local sources. It follows that all of the sources in the original network (b) is nonlocal, thereby witnessing its FQNN. }\label{1}
\end{figure}

  FNN is a crucial resource for device-independent quantum key distribution and quantum random number generation, offering security advantages over conventional protocols~\cite{PhysRevLett.120.020504,RevModPhys.89.015004}. Since its introduction, Bell-type inequalities for FQNN were soon verified experimentally~\cite{ar1,PhysRevLett.129.030502,PhysRevLett.130.190201,Wang2023}. Luo \textit{et al.} generalized the concept to \textit{hierarchical} network nonlocality~\cite{PhysRevA.110.022617} and proposed full quantum network nonlocality (FQNN), ruling out any correlations generated by the hybrid quantum and classical models. For a given integer $l$, a correlation is said to exhibit $l$-level quantum network nonlocality ($l$-QNN) if it cannot be reproduced when at least $l$ sources are local (the remaining sources may be arbitrarily nonlocal); otherwise it is $l$-level quantum network local ($l$-QNL). FQNN corresponds to the case $l=1$. Luo \textit{et al.} also proposed inequality criteria for detecting $l$-QNN in chain and star networks. Luo  \textit{et al.}  constructed a chained Bell inequality in networks for distinguishing different types of correlations \cite{Luo2024AdvQuantumTech}.

For networks with arbitrary structure, a natural approach is to decompose the whole network into star-shaped subnetworks and test each subnetwork for FQNN~\cite{PhysRevA.110.022617}. This decomposition strategy, however, requires many separate experimental runs and repeated measurements, quickly becoming prohibitive for large networks. Yang \textit{et al.} recently provided optimized inequalities for acyclic, cyclic, and general networks~\cite{Yang_2024}, but still no single-inequality test for FQNN in an arbitrary network exists.

In this work we present a general method that certifies FQNN in \textit{any} quantum network using only a single Bell-like inequality. The key idea is twofold. First, we derive an $l$-QNL inequality that holds for an arbitrary network and is tight enough to detect the presence of up to $w-1$ local sources, where $w$ is a topological parameter of the network. Second, we enlarge the original network by adding a suitable number of auxiliary local sources (and corresponding parties). If the enlarged network violates the $w$-QNN inequality, then it can contain at most $w-1$ local sources. Because we have explicitly inserted $w-1$ auxiliary local sources, all original sources must be nonlocal--hence the original network exhibits FQNN.
We illustrate the method with concrete examples (triangular, chain, cyclic, and tree-shaped networks) and show that the required inequality can be violated using Werner states and standard optimization techniques. Our scheme eliminates the need for repeated subnetwork testing and offers a practical path toward efficient FQNN certification in complex quantum networks.

This paper is structured as follows. Sec. II presents an inequality criterion for determining $l$-QNL, and uses the example of a triangular network to demonstrate that FQNN of the triangular network can be inferred by increasing local sources.  Sec. III builds upon this foundation to generalize the certification framework to chain networks, cyclic networks, and arbitrary quantum networks. Several concrete examples are also provided to demonstrate the feasibility of the scheme.  Sec. IV. presents the conclusion and discussion.

\section{Theoretical models for verifying FQNN with auxiliary local sources}

In large-scale networks, detecting their FQNN requires decomposing them into a vast number of star subnetworks, which greatly increases the complexity of the experimental process. This urgently raises the question of whether FQNN detection in an arbitrary network can be achieved by using a single Bell-like inequality only. We therefore reasoned about why it has not been possible to establish a Bell-like inequality criterion for FQNN detection in non-star networks. One possible reason is that star networks possess more independent parties in their structure. It then occurred to us that one could make a general network structurally closer to a star network by adding parties at its periphery. Ultimately, we discovered that the hierarchical network nonlocality of the newly constructed network (after adding local sources and parties) can witness the FQNN of the original non-star network. This idea is schematically illustrated in Fig. \ref{1}. Specifically, if we wish to verify that the correlations generated by the triangular quantum network shown in Fig. \ref{1}(b) exhibit FQNN, we may select any non-independent party within this network and introduce $2$ local sources to it (as illustrated in Fig. \ref{1}(c)). If the correlations produced by the network in Fig. \ref{1}(c) can be verified as $3$-QNN, and given that this network already contains $2$ local sources, it necessarily follows that all remaining sources in Fig. \ref{1}(c) are nonlocal. Consequently, this proves that the correlations generated by the original network in Fig. \ref{1}(c) exhibit FQNN.

Based on the preceding analysis, we first introduce an $l$-QNL inequality for arbitrary network structures. We begin by recalling that independent parties refer to spatially separated parties who do not share any common sources.
The independence number of a network being $h$ means that there exist $h$ independent parties in the network. Clearly, for any given quantum network, there always exists a maximum independence number and the corresponding set of maximum independent parties.

Based on the preceding analysis, we first introduce an $l$-QNL inequality for arbitrary network structures. We begin by recalling that independent parties refer to spatially separated parties who do not share any common sources.

Now let us observe an interesting phenomenon. A star network with $n$ parties has \(n-1\) sources and \(n-1\) maximal independent parties. If it contains a local source, then it must satisfy the following condition: (P) \textit{the network contains an independent party that receives only local sources.} However, this observation does not hold in non-star networks. For example, in triangular network shown in Fig. \ref{1}, the independent party is ${\bf A}_1$. Even if we know that there is one local source, it is still not guaranteed that (P) holds true. As illustrated in Fig. \ref{1}(a), the independent party ${\bf A}_1$ receives local and nonlocal states, and therefore ${\bf A}_1$ does not satisfy condition (P). But if there are  2 local sources in it, then (P) must be true.
Based on the characteristics of the sources received by the independent party, we can prove the following more general theorem.
Let
\begin{align}
w=\min \{v \ |\  &{\rm for\ a\ given\ network\ with\ an\ independent-party
\ partition,}\nonumber \\
 &{\rm (P)\  holds\ true \ if \ there \ are} \  v\ {\rm local\ sources}\}. \nonumber
\end{align}

\begin{theorem}\label{main} Let $\Xi(n,m)$ be an arbitrary general network consisting of $n$ parties ${\bf A}_1,{\bf A}_2,...,{\bf A}_n$ and $m$ sources. Its maximum independent number is $h$, and we label the corresponding parties as $\mathbf{A}_{i_1}, \mathbf{A}_{i_2}, \ldots, \mathbf{A}_{i_h}$. Each party ${\bf A}_{i}$ performs binary-input and binary-output measurements $A_{x_i=0}$ and $A_{x_i=1}$.
If the correlations generated by this network are $w$-QNL, then they satisfy the following inequality
\begin{align}\label{l=w(n)}
|I|^{\frac{1}{h}}+|J|^{\frac{1}{h}}\leq
2^{\frac{h-1}{2h}}
\end{align}
where $I=\langle \Pi_{i\in\Gamma}A^+_{x_i}\Pi_{j\in\bar{\Gamma}}A_{x_j=0}\rangle$ and
$J=\langle \Pi_{i\in\Gamma}A^-_{x_i}\Pi_{j\in\bar{\Gamma}}A_{x_j=1}\rangle$ with $A^\pm_{x_i}=\frac{1}{2}(A_{x_i=0}\pm A_{x_i=1})$. Here, $\Gamma=\{i_1,i_2,...,i_h\}$ and $\bar{\Gamma}=\{1,2,...,n\}\setminus\Gamma.$

\end{theorem}

\begin{proof} Given that at least one party among the independent parties receives only local sources, we select one such party and denote it as $\mathbf{A}_t$, where $t \in \Gamma$. Therefore, we can derive
\begin{align}
&|I|^{\frac{1}{h}}+|J|^{\frac{1}{h}} \nonumber \\
=&|\langle \prod_{i\in\Gamma}A^+_{x_i}\prod_{j\in\bar{\Gamma}}A_{x_j=0}\rangle|^{\frac{1}{h}}+|\langle \prod_{i\in\Gamma}A^-_{x_i}\prod_{j\in\bar{\Gamma}}A_{x_j=1}\rangle|^{\frac{1}{h}} \nonumber \\
\leq&|\langle A^+_{x_t}\rangle|^{\frac{1}{h}}
|\langle \prod_{i\in\Gamma\setminus t}A^+_{x_i}\prod_{j\in\bar{\Gamma}}A_{x_j=0}\rangle|^{\frac{1}{h}}  \nonumber \\
&+|\langle A^-_{x_t}\rangle|^{\frac{1}{h}}
|\langle \prod_{i\in\Gamma\setminus t}A^-_{x_i}\prod_{j\in\bar{\Gamma}}A_{x_j=1}\rangle|^{\frac{1}{h}}.
\end{align}
Clearly, by eliminating the independent party ${\bf A}_t$ from the network, the remaining network becomes ($h-1$)-independent.
 Ref. \cite{PhysRevLett.120.140402} provided a quantum bound for the ($h-1$)-independent network, leading to the following result:
\begin{align*}
&|\langle \prod_{i\in\Gamma\setminus t}A^+_{x_i}\prod_{j\in\bar{\Gamma}}A_{x_j=0}\rangle|^{\frac{1}{h-1}}+|\langle \prod_{i\in\Gamma\setminus t}A^-_{x_i}\prod_{j\in\bar{\Gamma}}A_{x_j=1}\rangle|^{\frac{1}{h-1}}\nonumber \\
&\leq\sqrt{2}.
\end{align*}
Therefore,
\begin{align}
|\langle \prod_{i\in\Gamma\setminus t}A^+_{x_i}\prod_{j\in\bar{\Gamma}}A_{x_j=0}\rangle|^{\frac{1}{h-1}}\leq\sqrt{2}{\rm cos}^2\theta,\\
|\langle \prod_{i\in\Gamma\setminus t}A^-_{x_i}\prod_{j\in\bar{\Gamma}}A_{x_j=1}\rangle|^{\frac{1}{h-1}}\leq\sqrt{2}{\rm sin}^2\theta,
\end{align}
 for some $\theta\in[0,\pi/2]$.
Furthermore, since the party ${\bf A}_t$ is independent, we have $|\langle A_{x_t}^+\rangle \pm \langle A_{x_t}^-\rangle|\leq1$. One can define $\langle A_{x_t}^+\rangle\leq {\rm cos}^2\nu$ and $\langle A_{x_t}^-\rangle\leq {\rm sin}^2\nu$ for some $\nu\in[0,\pi/2]$.
Based on the above discussion, we have
\begin{align}
&|I|^{\frac{1}{h}}+|J|^{\frac{1}{h}} \nonumber \\
\leq&2^{\frac{h-1}{2h}}[{\rm cos}^{\frac{2(h-1)}{h}}\theta {\rm cos}^{\frac{2}{h}}\nu+{\rm sin}^{\frac{2(h-1)}{h}}\theta {\rm sin}^{\frac{2}{h}}\nu] \nonumber \\
\leq&2^{\frac{h-1}{2h}}.
\end{align}
Here, the last inequality comes from  the fact that the binary continuous function $f(\theta,\nu)={\rm cos}^{\frac{2(h-1)}{h}}\theta {\rm cos}^{\frac{2}{h}}\nu+{\rm sin}^{\frac{2(h-1)}{h}}\theta {\rm sin}^{\frac{2}{h}}\nu$ has a unique extremum point, which is given by $\theta=\nu$.
\end{proof}

{\bf Remark. }
It is clear that the violation of Ineq. (\ref{l=w(n)}) implies that the number of local sources in the network is at most $w-1$.
Clearly, for any given network structure and partition of independent parties, the value of
$w$ can be uniquely determined.
Note that the number $w$ can be calculated using the following formula
\begin{align}\label{g3.6}
w = &\left\{
\begin{aligned}
&\sum_{i\in\Gamma}|\Lambda_i|-h+1,
&&\qquad \qquad\mbox{ for }\sum_{i\in\Gamma}|\Lambda_i|=m;\\
&\sum_{i_1,i_2\in\bar{\Gamma}}|\Lambda_{i_1}\cap\Lambda_{i_2}|+\sum_{i\in\Gamma}|\Lambda_i|-h+1,
&&\qquad\qquad \mbox{ for }\sum_{i\in\Gamma}|\Lambda_i|<m,
\end{aligned}
\right.
\end{align}
where $\Gamma=\{i_1,i_2,...,i_h\}$, $\bar{\Gamma}=\{1,2,\cdots,n\}\setminus\Gamma$,
and $|\Lambda_i|$ represent the number of elements in the set $\Lambda_i$.
In other words, for any given quantum network $\Xi(n,m)$, we can always use Eq. (\ref{g3.6}) to compute $w$ and construct the corresponding inequality criterion for detecting $w$-QNN based on Ineq. (\ref{l=w(n)}). If the inequality  can be violated, it always indicates that the network contains at most $w-1$ local states.

If it is known that there are already $w-1$ local sources in the network, then it can naturally be inferred that all remaining sources in the network are nonlocal. Therefore, based on this principle, we can modify the topological structure by increasing local sources and the corresponding parties (as shown in Fig. \ref{1} (c)), and infer that the network described in Fig. \ref{1} (b) generates the FQNN.
 Starting from this foundation,  we attempt to provide a method for determining whether the correlations generated by the network are of FQNN.  As an example, we infer the FQNN of the triangular network shown in Fig. \ref{1} (b) using this method.

\textit{Example 1}. (Small-size networks) We explain the proposed FQNN detection method using the case of a triangular network.  For the triangular network $\Xi(3,3)$ shown in Fig. \ref{1} (b), it consists of 3 parties and 3 sources. There exist at most 1 independent party, i.e., $h=1$.
Assume that ${\bf A}_1$ is the independent party.

To detect FQNN of network $\Xi(3,3)$, we introduce two auxiliary local sources and corresponding parties into the original network $\Xi(3,3)$, as  illustrated in Fig. \ref{1} (c).
The newly formed network consists of 5 parties, namely ${\bf A_1}, {\bf A_2}, {\bf A_3}, {\rm AUX_1},  {\rm AUX_2}$, and 5 sources, namely $S_1, S_2, S_3, S_4, S_5$.
For the sake of clarity, we denote the auxiliary parties ${\rm AUX_1}$ and  ${\rm AUX_2}$ as ${\bf A}_4$ and ${\bf A}_5$, respectively.
Among them, $S_4$ and $S_5$ are both auxiliary local sources.
We denote the newly formed network as $\Xi^{{\rm new}}(3,3)$.

For network $\Xi^{{\rm new}}(3,3)$, $n=5$ and $m=5$. There is at most 3 independent parties, i.e., $h=3$. $\Gamma=\{1,4,5\}$. From Eq. (\ref{g3.6}), we get $w=3.$ Thus, Ineq. (\ref{l=w(n)}) can be concretized as
\begin{align}\label{5node}
|I_1|^{\frac{1}{3}}+|J_1|^{\frac{1}{3}}\leq
2^{\frac{1}{3}}
\end{align}
where $I_1=\langle \Pi_{i\in\Gamma}A^+_{x_i}\Pi_{j\in\bar{\Gamma}}A_{x_j=0}\rangle$ and
$J_1=\langle \Pi_{i\in\Gamma}A^-_{x_i}\Pi_{j\in\bar{\Gamma}}A_{x_j=1}\rangle$ with $\Gamma=\{1,4,5\}$ and $\bar{\Gamma}=\{2,3\}$. If we choose the measurement settings for each party as
\begin{align}
&A_{x_1=0}=\sum_{r_1,r_2=0}^3\alpha^{10}_{r_1,r_2}\sigma_{r_1} \otimes\sigma_{r_2},  \\
&A_{x_1=1}=\sum_{r_1,r_2=0}^3\alpha^{11}_{r_1,r_2=0}\sigma_{r_1} \otimes\sigma_{r_2}, \\
&A_{x_2=0}=\sum_{u,v,w,l=0}^3\alpha^{20}_{u,v,w,l}\sigma_{u}\otimes\sigma_{v} \otimes \sigma_{w}\otimes\sigma_{l},  \\
&A_{x_2=1}=\sum_{u,v,w,l=0}^3\alpha^{21}_{u,v,w,l}\sigma_{u}\otimes\sigma_{v} \otimes \sigma_{w}\otimes\sigma_{l},  \\ &A_{x_3=0}=\sum_{i,j=0}^3\alpha^{30}_{i,j}\sigma_{i} \otimes\sigma_{j},  \\
&A_{x_3=1}=\sum_{i,j=0}^3\alpha^{31}_{i,j}\sigma_{i} \otimes\sigma_{j}, \\
&A_{x_t=0}=\sum_{r_3=0}^3\alpha^{t0}_{r_3}\sigma_{r_3},\,\,\,A_{x_t=1}=\sum_{r_3=0}^3\alpha^{t1}_{r_3}\sigma_{r_3}
\end{align}
for any $t\in\{4,5\}$. For any observable, it is always required that the absolute values of its eigenvalues are at most 1.  Assume that each source in the network $\Xi^{\rm new}(3,3)$ is a Werner state, given by $S_i=p_i|\psi^-\rangle\langle\psi^-|+\frac{1-p_i}{4}I_2\otimes I_2$, where
  $|\psi^-\rangle=\frac{1}{\sqrt{2}}(|01\rangle-|10\rangle).$
 Then, we get
\begin{align}
I_{1}=&\sum_{i,j,u,v,w,l=0}^3\alpha^{20}_{u,j,w,l}
\alpha^{30}_{i,j}(\alpha_{i,u}^{10}+\alpha_{i,u}^{11})(\alpha_{w}^{40}+\alpha_{w}^{41})\nonumber \\
&\times
(\alpha_{l}^{50}+\alpha_{l}^{51})(-p_1)^i(-p_2)^u(-p_3)^j(-p_4)^w(-p_5)^l,\\
J_{1}=&\sum_{i,j,u,v,w,l=0}^3\alpha^{21}_{u,j,w,l}
\alpha^{31}_{i,j}(\alpha_{i,u}^{10}-\alpha_{i,u}^{11})(\alpha_{w}^{40}-\alpha_{w}^{41})\nonumber \\
&\times
(\alpha_{l}^{50}-\alpha_{l}^{51})(-p_1)^i(-p_2)^u(-p_3)^j(-p_4)^w(-p_5)^l.\end{align}
Here,  $(-p)^k =(-p)^{1-\delta(k,0)}$.
Given $p_1=p_2=p_3=0.9$ and $p_4=p_5=\frac{1}{3}$, we employ the Sequential Least Squares Quadratic Programming (SLSQP) algorithm to maximize the left-hand side  of Ineq. (\ref{5node}). The optimization yields a maximum value of $|I_1|^{1/3}+|J_1|^{1/3} = 1.357$, which demonstrates that Ineq. (\ref{5node}) can be violated. That is to say, if the network contains local sources, one can always employ the SLSQP algorithm to find suitable measurements that violate the inequality, thereby demonstrating its validity. This violation implies that the network contains at most two local sources.

Considering the current distribution of sources in the network, where it is known that two local sources already exist, we can conclude that all remaining sources in the triangular network are nonlocal. This means that we can determine the correlations generated by the network $\Xi(3,3)$ are FQNN using the violation of Ineq. (\ref{5node}).

Subsequently, we extend the methods presented in this section to more general network structures.
In the following figures,
the yellow circles represent independent parties.
Auxiliary local  sources are represented by red wavy lines, and ${\rm AUX}_i$ denotes the newly added auxiliary party.

\section{General networks}
In this section, we show that the scheme proposed in Section II is applicable to general networks.

\subsection{Chain networks}

Based on the discussion in Section II, we witness the FQNN of a chain network via the following two propositions.

For a chain network $\Xi(n,n-1)$ with $n$ parties and $n-1$ sources, the maximum number of independent parties is $\frac{n}{2}$, and all parties with odd indices are independent. By adding $\frac{n}{2}-1$ auxiliary local sources to party ${\bf A}_i$ ($i\in\{2,4,...,n\}$) and its connected parties, we construct a new network $\Xi^{\rm new}(n,n-1)$. We then obtain from Theorem \ref{main} the following results.

\begin{corollary}
The network $\Xi(n,n-1)$ (even $n$) is FQNN if there are  measurements in $\Xi^{\rm new}(n,n-1)$ such that the generated correlations violate the inequality
\begin{align}
|I_{\rm chae}|^{\frac{1}{n-1}}+|J_{\rm chae}|^{\frac{1}{n-1}}\leq2^{\frac{n-2}{2n-2}},
\end{align}
where $I_{\rm chae}=\langle \Pi_{i\in\Gamma^{\rm new}}A^+_{x_i}\Pi_{j\in\bar{\Gamma}^{\rm new}}A_{x_j=0}\rangle$ and
$J_{\rm chae}=\langle \Pi_{i\in\Gamma^{\rm new}}A^-_{x_i}\Pi_{j\in\bar{\Gamma}^{\rm new}}A_{x_j=1}\rangle$ with $\Gamma^{\rm new}=\{1,3,...,n-1,n+1,\ldots,\frac{3n}{2}-1\}$ and  $\bar{\Gamma}^{\rm new}=\{2,4,...,n\}$.
\end{corollary}

When $n$ is odd the maximum number of independent parties chain network $\Xi(n,n-1)$ is $\frac{n+1}{2}$, and all parties with odd indices are independent. By adding $\frac{1}{2}(n-3)$ auxiliary local sources to party ${\bf A}_i$ ($i\in\{2,4,...,n-1\}$) and its connected parties, we construct a new network $\Xi^{\rm new}(n,n-1)$. We then obtain another result.

\begin{corollary}
The network $\Xi(n,n-1)$ (odd $n$) is FQNN if there are  measurements in $\Xi^{\rm new}(n,n-1)$ such that the generated correlations violate the inequality
\begin{align}
|I_{\rm chao}|^{\frac{1}{n-1}}+|J_{\rm chao}|^{\frac{1}{n-1}}\leq2^{\frac{n-2}{2n-2}},
\end{align}
where $I_{\rm chao}=\langle \Pi_{i\in\Gamma^{\rm new}}A^+_{x_i}\Pi_{j\in\bar{\Gamma}^{\rm new}}A_{x_j=0}\rangle$ and
$J_{\rm chao}=\langle \Pi_{i\in\Gamma^{\rm new}}A^-_{x_i}\Pi_{j\in\bar{\Gamma}^{\rm new}}A_{x_j=1}\rangle$ with $\Gamma^{\rm new}=\{1,3,...,n,n+1,\ldots,\frac{3n-3}{2}\}$ and  $\bar{\Gamma}^{\rm new}=\{2,4,...,n-1\}$.
\end{corollary}

\textit{Example 2}. We certify whether the correlations produced by the network in Fig. \ref{chain} (a) are FQNN.
This is achieved by augmenting the original network with an auxiliary local source and an additional party. Specifically, as shown in Fig.
\ref{chain} (b), a local source and a party $\text{AUX}_1$ are appended to the configuration in Fig. \ref{chain} (a).
For simplicity, let ${\rm AUX}_1$ denote ${\bf A}_6$. It follows from Eq. (\ref{g3.6}) that we get $w=2$. If the correlations generated by this network are $2$-QNL, according to Ineq. (\ref{l=w(n)}), we can derive
 \begin{eqnarray}\label{ch5}
|I_2|^{\frac{1}{4}}+|J_2|^{\frac{1}{4}}\leq
2^{\frac{3}{8}}
\end{eqnarray}
where  $I_2=\langle A_{x_1}^+A_{{x_2}=0}A_{x_3}^+A_{{x_4}=0}A_{x_5}^+A_{x_6}^+\rangle$
and
$J_2=\langle A_{x_1}^-A_{{x_2}=1}A_{x_3}^-A_{{x_4}=1}A_{x_5}^-A_{x_6}^-\rangle$.
Here, $A_{x_i=0}$ and $A_{x_i=1}$ are observables for party $\mathbf{A}_i$, for all $i \in {1,2,\ldots,6}$.
In this case, the measurement schemes implemented by each party are considered to be
\begin{align}
&A_{x_s=0}=\sum_{r_1=0}^3\alpha^{s0}_{r_1}\sigma_{r_1},\,\,\,A_{x_s=1}=\sum_{r_1=0}^3\alpha^{s1}_{r_1}\sigma_{r_1},\\
&A_{x_2=0}=\sum_{r_1,r_2=0}^3\alpha^{20}_{r_1,r_2}\sigma_{r_1} \otimes\sigma_{r_2},  \\
&A_{x_2=1}=\sum_{r_1,r_2=0}^3\alpha^{21}_{r_1,r_2}\sigma_{r_1} \otimes\sigma_{r_2}, \\
&A_{x_3=0}=\sum_{i,j=0}^3\alpha^{30}_{i,j}\sigma_{i} \otimes\sigma_{j},  \\
&A_{x_3=1}=\sum_{i,j=0}^3\alpha^{31}_{i,j}\sigma_{i} \otimes\sigma_{j}, \\
&A_{x_4=0}=\sum_{u,v,w=0}^3\alpha^{40}_{u,v,w}\sigma_{u}\otimes\sigma_{v} \otimes \sigma_{w},  \\
&A_{x_4=1}=\sum_{u,v,w=0}^3\alpha^{41}_{u,v,w}\sigma_{u}\otimes\sigma_{v} \otimes \sigma_{w},
\end{align}
for any $s\in\{1,5,6\}$. Here, the measurement parameters must be chosen such that the maximum absolute value of the eigenvalues of the observable is less than or equal to 1. Then we get
\begin{align}
I_2=&\sum_{i,j,u,v,w=0}^3\alpha_{i,j}^{20}\alpha_{u,v,w}^{40}(\alpha_{i}^{10}+\alpha_i^{11})(\alpha^{30}_{ju}+\alpha^{31}_{ju}) \nonumber \\&\times(\alpha_{v}^{50}+\alpha_{v}^{51})(\alpha_{w}^{60}+\alpha_{w}^{61})(-p_1)^i(-p_2)^j\nonumber \\
&\times(-p_3)^u(-p_4)^v(-p_5)^w,\\
J_2=&\sum_{i,j,u,v,w=0}^3\alpha_{i,j}^{21}\alpha_{u,v,w}^{41}(\alpha_{i}^{10}-\alpha_i^{11})(\alpha^{30}_{ju}-\alpha^{31}_{ju}) \nonumber \\&\times(\alpha_{v}^{50}-\alpha_{v}^{51})(\alpha_{w}^{60}-\alpha_{w}^{61})(-p_1)^i(-p_2)^j \nonumber \\
&\times(-p_3)^u(-p_4)^v(-p_5)^w.
\end{align}
Here,  $(-p)^k =(-p)^{1-\delta(k,0)}$.
Take $p_1=p_2=p_3=p_4=0.9$ and $p_5=\frac{1}{3}$, we employ SLSQP algorithm to maximize the left-hand side  of Ineq. (\ref{ch5}). The optimization yields a maximum value of $|I_2|^{1/4}+|J_2|^{1/4} =1.315$, which demonstrates that Ineq. (\ref{ch5}) can be violated.

The violation of Ineq. (\ref{ch5}) implies that the network contains at most one local source. In fact,  there is already a local source present in the network. Using this fact, we conclude that all sources $S_1, S_2, S_3, S_4$ in the network are nonlocal and thus confirm that the generated distribution corresponds to the FQNN in Fig. \ref{chain} (a).

\begin{figure}
  \centering
\subfigure[]{\includegraphics[width=3in]{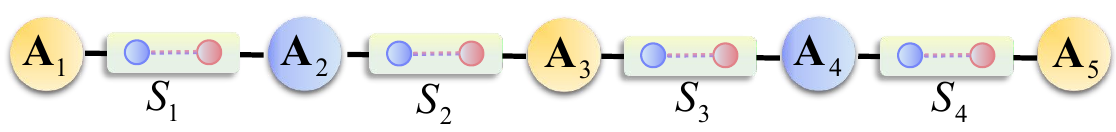}}
  \subfigure[]{\includegraphics[width=3in]{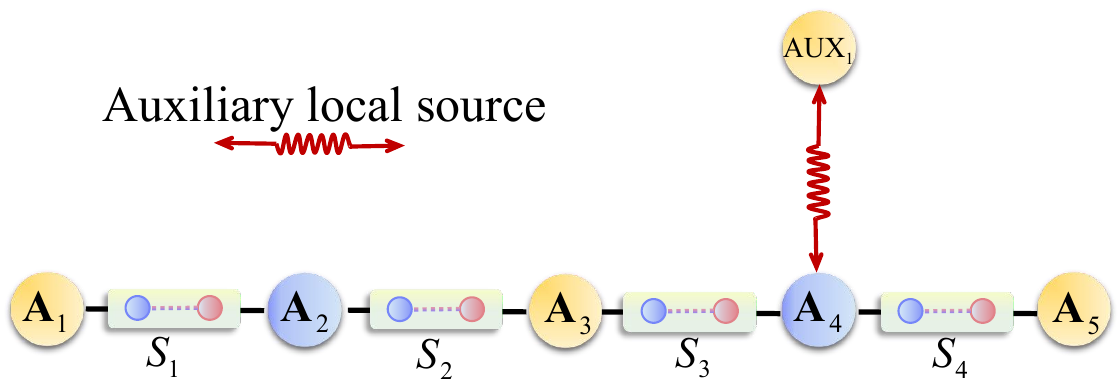}}
  \caption{ (a) Chain network in which parties ${\bf A}_{i}$ and ${\bf A}_{i+1}$ share an independent source $S_i$ for any $i\in\{1,2,3,4\}$. (b) A new network is formed by adding an auxiliary local source.
  }\label{chain}
\end{figure}

\subsection{Cyclic networks}

For a cyclic network $\Xi(n,n)$ comprising $n$ parties and $n$ sources, the presence of FQNN in the generated correlations can still be detected by introducing auxiliary local sources. It follows directly that the maximum independent number of the network is $\frac{n}{2}$ when $n$ is even, and $\frac{n-1}{2}$ when $n$ is odd (see Fig. \ref{cyclicn} (a)). We first consider the case where $n$ is even in the following discussion. To detect whether the correlations generated by network $\Xi(n,n)$ are of FQNN, we mainly consider the following steps:

We first identify the parties corresponding to the maximum independence number in the network, and denote the set of their indices as $\Gamma$. It is evident that when $n$ is even, the parties with even indices are all independent. Thus, $\Gamma=\{2,4,\ldots,n\}$. Denote the remaining parties as $\bar{\Gamma}=\{1,2,\ldots,n\}\backslash\Gamma$.

If $v\in\bar{\Gamma}$, then add $\frac{n}{2}$ auxiliary local sources and the corresponding auxiliary parties to party ${\bf A}_v$. In Fig. \ref{cyclicn} (b), $v=5$ and $k=\frac{n}{2}$. Denote the network with the added auxiliary local sources and auxiliary parties as $\Xi^{\rm new}(n,n)$

For the network $\Xi^{\rm new}(n,n)$, we can label each newly added auxiliary party as ${\bf A}_{n+1},\ldots,{\bf A}_{\frac{3n}{2}}$. This implies that the maximum independence number of the network is $n$, and $w(n)=\frac{n}{2}+1$. The set of indices of the parties corresponding to the maximum independence number is denoted as $\Gamma^{\rm new}=\{2,4,\ldots,n,n+1,\ldots,\frac{3n}{2}\}$, and the set of indices of the remaining parties is denoted as $\bar{\Gamma}^{\rm new}=\bar{\Gamma}=\{1,3,...,n-1\}$. Moreover, it follows from Ineq. (\ref{l=w(n)}) that we can  construct the following inequality
\begin{align}\label{cn}
|I_{\rm cyco}|^{\frac{1}{n}}+|J_{\rm cyco}|^{\frac{1}{n}}\leq
2^{\frac{n-1}{2n}}
\end{align}
where $I_{\rm cyco}=\langle \Pi_{i\in\Gamma^{\rm new}}A^+_{x_i}\Pi_{j\in\bar{\Gamma}^{\rm new}}A_{x_j=0}\rangle$ and
$J_{\rm cyco}=\langle \Pi_{i\in\Gamma^{\rm new}}A^-_{x_i}\Pi_{j\in\bar{\Gamma}^{\rm new}}A_{x_j=1}\rangle$.

We now choose the optimal measurements to compute the left-hand side of Ineq. (\ref{cn}). Violation of Ineq. (\ref{cn})
means that there exists at most $\frac{n}{2}$ local sources. Given that there are already $\frac{n}{2}$ local sources in the current network, it can be inferred that the correlations generated by network $\Xi(n,n)$ are of FQNN.

\begin{figure}
  \centering
\subfigure[]{\includegraphics[width=1.6in]{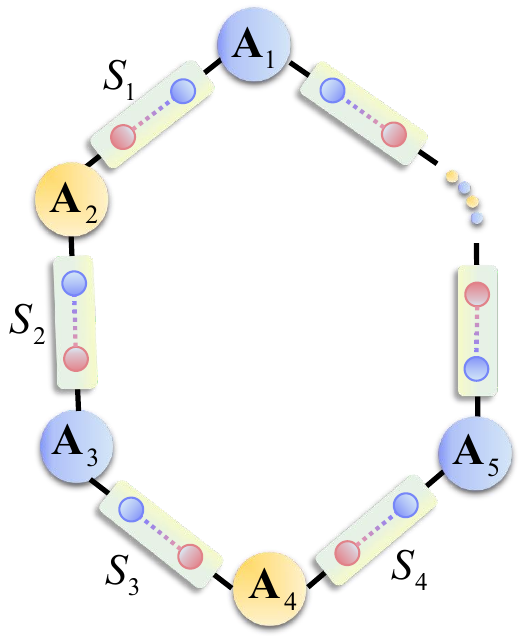}}
  \subfigure[]{\includegraphics[width=2.3in]{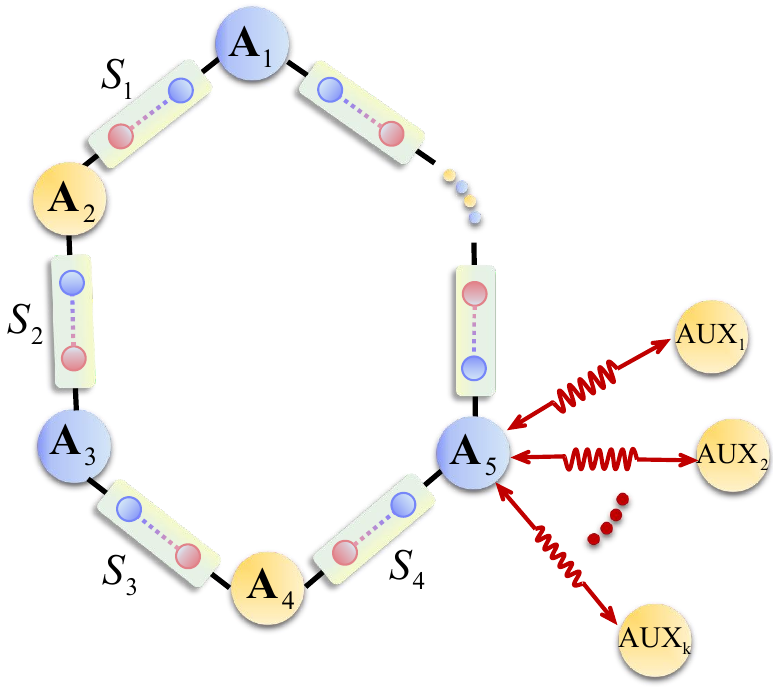}}
  \caption{
  (a) A cyclic network $\Xi(n,n)$ composed of $n$ parties and $n$ sources. (b) The network $\Xi^{\rm new}(n,n)$, which is formed by adding $k$ auxiliary local sources and  $k$ auxiliary parties (${\rm AUX}_1,{\rm AUX}_2,\ldots, {\rm AUX_k}$) to a cyclic network $\Xi(n,n)$.
   }\label{cyclicn}
\end{figure}

To summarize, for the cyclic networks $\Xi(n,n)$,  the maximum number of independent parties is $\frac{n}{2}$, and all parties with even indices are independent. By adding $\frac{n}{2}$ auxiliary local sources to party ${\bf A}_i$ ($i\in\{1,3,...,n-1\}$) and its connected parties, we construct a new network $\Xi^{\rm new}(n,n)$. We then obtain the following results.

\begin{corollary}
The network $\Xi(n,n)$ is FQNN if there are suitable measurements in $\Xi^{\rm new}(n,n)$ such that the generated correlations violate the inequality
\begin{align}
|I_{\rm cyce}|^{\frac{1}{n}}+|J_{\rm cyce}|^{\frac{1}{n}}\leq2^{\frac{n-1}{2n}},
\end{align}
where $I_{\rm cyce}=\langle \Pi_{i\in\Gamma^{\rm new}}A^+_{x_i}\Pi_{j\in\bar{\Gamma}^{\rm new}}A_{x_j=0}\rangle$ and
$J_{\rm cyce}=\langle \Pi_{i\in\Gamma^{\rm new}}A^-_{x_i}\Pi_{j\in\bar{\Gamma}^{\rm new}}A_{x_j=1}\rangle$ with $\Gamma^{\rm new}=\{2,4,\ldots,n,n+1,\ldots,\frac{3n}{2}\}$ and  $\bar{\Gamma}^{\rm new}=\{1,3,...,n-1\}$.
\end{corollary}

When $n$ is odd, we can also use a similar method to determine whether the correlations generated by network $\Xi(n,n)$ are of FQNN by constructing inequalities.
The only difference is that the number of additional auxiliary local sources and parties changes from the original $\frac{n}{2}$ to $\frac{n+1}{2}$.

For the cyclic networks $\Xi(n,n)$,  the maximum number of independent parties is $\frac{n-1}{2}$, and all parties with even indices are independent. By adding $\frac{n+1}{2}$ auxiliary local sources to party ${\bf A}_i$ ($i\in\{1,3,...,n\}$) and its connected parties, we construct a new network $\Xi^{\rm new}(n,n)$. For this extended network we have:

\begin{corollary}
The network $\Xi(n,n)$ is FQNN if there are measurements in $\Xi^{\rm new}(n,n)$ such that the generated correlations violate the inequality
\begin{align}
|I_{\rm cyco}|^{\frac{1}{n}}+|J_{\rm cyco}|^{\frac{1}{n}}\leq2^{\frac{n-1}{2n}},
\end{align}
where $I_{\rm cyco}=\langle \Pi_{i\in\Gamma^{\rm new}}A^+_{x_i}\Pi_{j\in\bar{\Gamma}^{\rm new}}A_{x_j=0}\rangle$ and
$J_{\rm cyco}=\langle \Pi_{i\in\Gamma^{\rm new}}A^-_{x_i}\Pi_{j\in\bar{\Gamma}^{\rm new}}A_{x_j=1}\rangle$ with $\Gamma^{\rm new}=\{2,4,\ldots,n,n+1,\ldots,\frac{3n+1}{2}\}$ and  $\bar{\Gamma}^{\rm new}=\{1,3,...,n\}$.
\end{corollary}

\begin{figure}
  \centering
\subfigure[]{\includegraphics[width=1.5in]{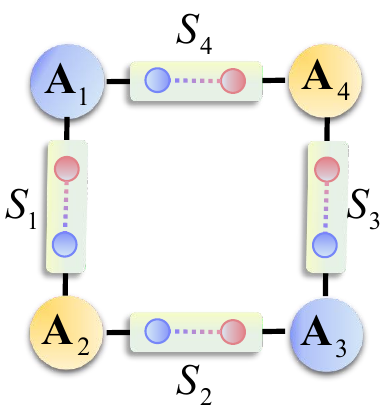}}
  \subfigure[]{\includegraphics[width=2.5in]{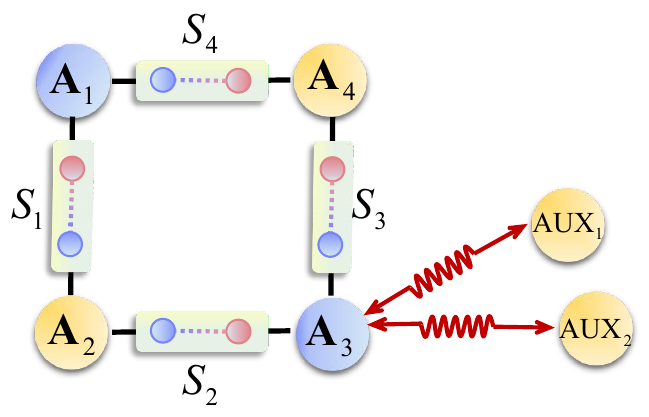}}
  \caption{ The yellow circles in the figure represent independent parties.  (a) A cyclic network $\Xi(4,4)$.
  (b) The network $\Xi^{\rm new}(4,4)$, which is formed by adding 2 auxiliary local sources and 2 auxiliary parties to a cyclic $\Xi(4,4)$. }\label{cyclic1}
\end{figure}

\textit{Example 3}.  For a cyclic network $\Xi(4,4)$ shown in Fig. \ref{cyclic1} (a) $\Xi(4,4)$. By selecting an arbitrary party from the set of non-independent parties in each network and adjoining several auxiliary local sources together with a corresponding auxiliary party at that location, we construct modified networks. Specifically, when 2 local sources are added, the resulting networks are illustrated in Fig. \ref{cyclic1} (b), which we denote as $\Xi^{\text{new}}(4,4)$.

For the network $\Xi^{\rm new}(4,4)$, denote the two auxiliary parties as ${\bf A}_5$ and ${\bf A}_6$. It follows from Eq. (\ref{g3.6}) that we get $w=3$. Thus, if the correlations generated by this network are $3$-QNL, then
Ineq. (\ref{l=w(n)}) can be rewritten as
\begin{align}\label{cyclic4}
|I_3|^{\frac{1}{4}}+|J_3|^{\frac{1}{4}}\leq
2^{\frac{3}{8}},
\end{align}
where $I_3=\langle \Pi_{i\in\Gamma}A^+_{x_i}\Pi_{j\in\bar{\Gamma}}A_{x_j=0}\rangle$ and
$J_3=\langle \Pi_{i\in\Gamma}A^-_{x_i}\Pi_{j\in\bar{\Gamma}}A_{x_j=1}\rangle$ with $\Gamma=\{2,4,5,6\}$ and $\bar{\Gamma}=\{1,3\}$.  Similarly to the approach in Example 1, one can use SLSQP algorithm to find a measurement scheme that violates the inequality. Violation of Ineq. (\ref{cyclic4}) implies that there exist at most 2 local sources in the network. Given that two local sources already exist in the network, it can be inferred that all sources in network $\Xi(4,4)$ are nonlocal. This implies that the correlations generated by network $\Xi(4,4)$ are of FQNN.

\subsection{Networks with arbitrary configuration}

Consider an arbitrary quantum network denoted as $\Xi(n,m)$, which consists of $n$ parties ${\bf A}_1, {\bf A}_2$, \ldots, ${\bf A}_n$ and $m$ sources $S_1,S_2,\ldots,S_m$. To obtain the inequality criteria for determining the FQNN of a network $\Xi(n,m)$, we need to construct the corresponding inequality criteria through the following four steps.

Suppose the maximum independence number of $\Xi(n,m)$ is $h$. Then, the parties that constitute this maximum independent set are denoted as ${\bf A}_{i_1}, {\bf A}_{i_2},\ldots, {\bf A}_{i_h}$. We define  $\Gamma=\{i_1,i_2,\ldots,i_h\}$ as the set of their indices, and let $\bar{\Gamma}=\{1,2,\ldots,n\}\backslash\Gamma$ be the set of indices for the remaining parties.

We then select an element from the set $\bar{\Gamma}$. If $r\in \bar{\Gamma}$, add the $N(n,m)$ auxiliary local sources and auxiliary parties to the party ${\bf A}_r$, where $N(n,m)=\sum_{i_1,i_2\in\bar{\Gamma}}|\Lambda_{i_1}\cap\Lambda_{i_2}|+\sum_{i\in\Gamma}|\Lambda_i|-h$. We denote the resulting network as $\Xi^{\rm new}(n,m)$. The schematic diagram is shown in Fig. \ref{shiyi}, where local sources are added exclusively to non-independent parties. Clearly, the maximum independence number of $\Xi^{\rm new}(n,m)$ is $h+N(n,m)$.

For the network $\Xi^{\rm new}(n,m)$,
label each newly added auxiliary party as ${\bf A}_{n+1},{\bf A}_{n+2},\ldots,{\bf A}_{n+N(n,m)}$. According to Eq. (\ref{g3.6}), we can get $w=\sum_{i_1,i_2\in\bar{\Gamma}}|\Lambda_{i_1}\cap\Lambda_{i_2}|+\sum_{i\in\Gamma}|\Lambda_i|-h+1$. Ineq. (\ref{l=w(n)}) can be rewritten as
\begin{align}\label{gn}
|I_{\rm gen}|^{\frac{1}{h+N(n,m)}}+|J_{\rm gen}|^{\frac{1}{h+N(n,m)}}\leq
2^{\frac{h+N(n,m)-1}{2(h+N(n,m))}}
\end{align}
where $I_{\rm gen}=\langle \Pi_{i\in\Gamma^{\rm new}}A^+_{x_i}\Pi_{j\in\bar{\Gamma}^{\rm new}}A_{x_j=0}\rangle$ and
$J_{\rm gen}=\langle \Pi_{i\in\Gamma^{\rm new}}A^-_{x_i}\Pi_{j\in\bar{\Gamma}^{\rm new}}A_{x_j=1}\rangle$ with $\Gamma^{\rm new}=\Gamma\cup\{n+1,n+2,\ldots,n+N(n,m)\}$ and $\bar{\Gamma}^{\rm new}=\bar{\Gamma}$.

Finally, we choose the optimal measurements to compute the left-hand side of Ineq. (\ref{gn}).
Due to the complexity of the network structure, the measurements performed by some parties are relatively complex. Therefore, we can employ SLSQP algorithm to maximize the value on the left-hand side of Ineq. (\ref{gn}), with the aim of violating the inequality for given source.
Violation of Ineq. (\ref{gn})
means that that the number of local sources in the network is at most $w-1$. It is worth noting that $w(n)-1=N(n,m)$.
Since the network already contains $N(n,m)$ local sources, the correlations generated by $\Xi(n,m)$ are therefore FQNN.

\begin{figure}
  \centering
  {\includegraphics[width=2.3
  in]{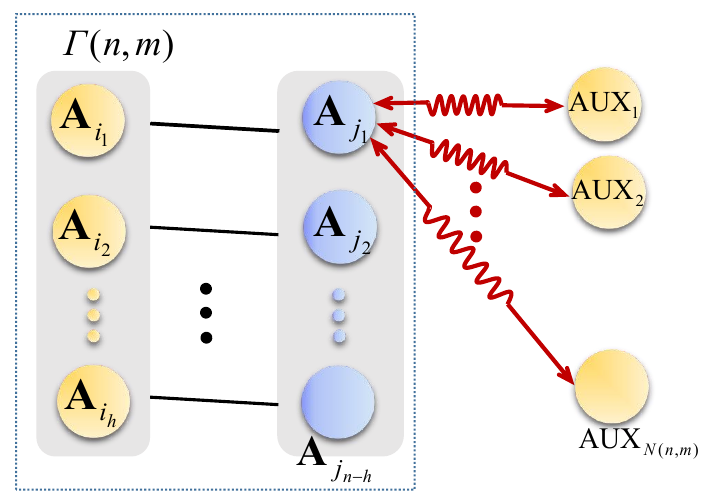}}
  \caption{  ${\bf A}_{i_1}, {\bf A}_{i_2},...,{\bf A}_{i_h}$ and ${\bf A}_{j_1}, {\bf A}_{j_2},...,{\bf A}_{j_{n-h}}$ represent independent parties and non-independent parties, respectively, with the black lines indicating the sources shared among the parties. The schematic illustrates the newly formed network after adding $N(n,m)$ additional local sources and parties to the original network.
   }\label{shiyi}
\end{figure}

In summary, for a  network $\Xi(n,m)$, suppose that  the maximum number of independent parties is $h$ with $\Gamma=\{i_1,i_2,...,i_h\}$ and $\bar{\Gamma}=\{1,2,\ldots,n\}\backslash\Gamma$. By adding $N(n,m)=\sum_{i_1,i_2\in\bar{\Gamma}}|\Lambda_{i_1}\cap\Lambda_{i_2}|+\sum_{i\in\Gamma}|\Lambda_i|-h$ auxiliary local sources to party ${\bf A}_i$ ($i\in\bar{\Gamma}$) and its connected parties shown in Fig. \ref{shiyi}, we construct a new network $\Xi^{\rm new}(n,m)$. We then propose the following corollary.

\begin{corollary}
The network $\Xi(n,m)$ is FQNN if there are suitable measurements in $\Xi^{\rm new}(n,m)$ such that the generated correlations violate the inequality
\begin{align}\label{gen}
|I_{\rm gen}|^{\frac{1}{h+N(n,m)}}+|J_{\rm gen}|^{\frac{1}{h+N(n,m)}}\leq2^{\frac{h+N(n,m)-1}{2(h+N(n,m))}},
\end{align}
where $I_{\rm gen}=\langle \Pi_{i\in\Gamma^{\rm new}}A^+_{x_i}\Pi_{j\in\bar{\Gamma}^{\rm new}}A_{x_j=0}\rangle$ and
$J_{\rm gen}=\langle \Pi_{i\in\Gamma^{\rm new}}A^-_{x_i}\Pi_{j\in\bar{\Gamma}^{\rm new}}A_{x_j=1}\rangle$ with $\Gamma^{\rm new}=\Gamma\cup\{n+1,n+2,\ldots,n+N(n,m)\}$ and $\bar{\Gamma}^{\rm new}=\bar{\Gamma}$.
\end{corollary}

Based on the above discussion, for any given network $\Xi(n,m)$, a new network $\Xi^{\rm new}(n,m)$ consisting of additional sources and independent parties can be constructed from its structure by following Steps 1-3. This construction allows an inequality to be formulated based on Ineq. (\ref{gen}). Before using SLSQP algorithm, we first need to clarify the form of measurements performed by the parties and the expression on the left-hand side of Ineq. (\ref{gen}). For any party $\mathbf{A}_r$ receiving $t$ hidden variables (i.e., $|\Lambda_r|=t$), it is natural to assume its measurement strategy takes the specific form:
\begin{align}
A_{x_r=0}= \sum_{r_1,r_2,\ldots,r_t=0}^3 \alpha^{r0}_{r_1,r_2,\ldots,r_t} \sigma_{r_1} \otimes\sigma_{r_2} \otimes \cdots \otimes \sigma_{r_t}, \\
A_{x_r=1}= \sum_{r_1,r_2\ldots,r_t=0}^3 \alpha^{r1}_{r_1,r_2,\ldots,r_t} \sigma_{r_1} \otimes\sigma_{r_2} \otimes \cdots \otimes \sigma_{r_t},
\end{align}
where $\sigma_0 = I_2$, $\sigma_1 = \sigma_x$, $\sigma_2 = \sigma_y$, and $\sigma_3 = \sigma_z$. These observables are constrained by $|A_{x_r=0}| \leq 1$ and $|A_{x_r=1}| \leq 1$.
Now, consider the specific case where all sources $S_i$ in the network $\Xi^{\rm new}(n,m)$ are Werner states with parameters $p_i$.   If a source is an additionally introduced local source, its state parameter is taken to be $\frac{1}{3}$. Under this setting, if $\Lambda_{j_1}\cap\Lambda_{j_2}=\emptyset$ for any $j_1, j_2\in\bar{\Gamma}^{\rm new}$, the following results can always be computed:
\begin{align}
I_{\rm new}=&\Pi_{i\in\Gamma}(\alpha^{i0}_{i_1,i_2,...,i_{|\Lambda_i|}}+\alpha^{i1}_{i_1,i_2,...,i_{|\Lambda_i|}})\Pi_{j\in\bar{\Gamma}}\alpha^{j0}_{j_1,j_2,...,j_{|\Lambda_j|}}\nonumber \\
&\times\Pi_{j\in\bar{\Gamma}}(-p_1)^{j_1}(-p_2)^{j_2}\cdots(-p_{m+N(n,m)})^{j_{|\Lambda_j|}}, \label{new1}\\
J_{\rm new}=&\Pi_{i\in\Gamma}(\alpha^{i0}_{i_1,i_2,...,i_{|\Lambda_i|}}-\alpha^{i1}_{i_1,i_2,...,i_{|\Lambda_i|}})\Pi_{j\in\bar{\Gamma}}\alpha^{j1}_{j_1,j_2,...,j_{|\Lambda_j|}}\nonumber \\
&\times\Pi_{j\in\bar{\Gamma}}(-p_1)^{j_1}(-p_2)^{j_2}\cdots(-p_{m+N(n,m)})^{j_{|\Lambda_j|}}. \label{new2}
\end{align}
Here, the exponent is interpreted as follows: $(-p_y)^{x}=(-p_y)^{1-\delta(k,0)}$.
Using Eqs. (\ref{new1}) and (\ref{new2}), we can then employ the SLSQP algorithm to find the optimal measurements for the network, with the aim of violating Ineq. (\ref{gn}) even in the presence of local sources. If non-independent parties share a common source, then Eqs. (\ref{new1}) and (\ref{new2}) need to be modified according to the network structure.
We illustrate this with the following concrete examples.

\begin{figure}
  \centering
\subfigure[]{\includegraphics[width=1.8in]{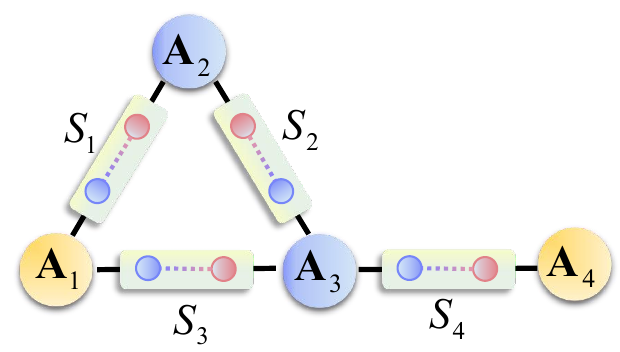}}
  \subfigure[]{\includegraphics[width=2.3
  in]{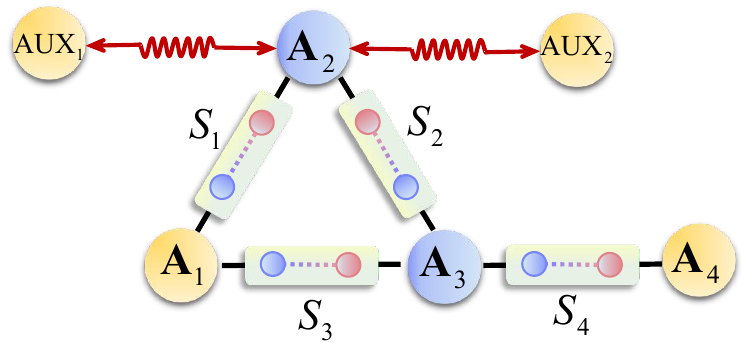}}
  \caption{  (a) A network $\Xi(4,4)$ with the maximum independence number 2. (b) The new network $\Xi^{\rm new}(4,4)$ formed after adding two auxiliary local sources and two auxiliary parties. 
   }\label{ex}
\end{figure}

\textit{Example 4}. For the network $\Xi(4,4)$ described in Fig. \ref{ex}(a),  its maximum independence number is 2 $(h=2)$, and the corresponding independent parties are ${\bf A}_1,{\bf A}_4$. Therefore, we have $\Gamma_1=\{1,4\}$ and $\bar{\Gamma}_1=\{2,3\}$.


By calculation, we can obtain $N(4,4)=1+3-2=2.$ Since $2\in\bar{\Gamma}_1$, we can add two auxiliary local sources and two parties to the party ${\bf A}_2$, forming the network $\Xi^{\rm new}(4,4)$ as shown in Fig. \ref{ex} (b). It is evident that the maximum independent number of network $\Xi^{\rm new}(4,4)$ is $h+N(4,4)=4$.
Then, we denote the added auxiliary parties as ${\bf A}_5$, ${\bf A}_{6}$, and we get $w=3$. Ineq. (\ref{gn}) can be expressed as
\begin{align}\label{ex1}
|I_{\rm new 1}|^{\frac{1}{4}}+|J_{\rm new 1}|^{\frac{1}{4}}\leq
2^{\frac{3}{8}}
\end{align}
where $I_{\rm new 1}=\langle \Pi_{i\in\Gamma^{\rm new1}}A^+_{x_i}\Pi_{j\in\bar{\Gamma}^{\rm new1}}A_{x_j=0}\rangle$ and
$J_{\rm new 1}=\langle \Pi_{i\in\Gamma^{\rm new1}}A^-_{x_i}\Pi_{j\in\bar{\Gamma}^{\rm new1}}A_{x_j=1}\rangle$ with $\Gamma^{\rm new1}=\Gamma_1\cup\{5,6\}$ and $\bar{\Gamma}^{\rm new1}=\bar{\Gamma}_1=\{2,3\}$.

Assume that each source in the network $\Xi^{\rm new}(4,4)$ is a Werner state, given by $S_i=p_i|\psi^-\rangle\langle\psi^-|+\frac{1-p_i}{4}I_2\otimes I_2$, where
  $|\psi^-\rangle=\frac{1}{\sqrt{2}}(|01\rangle-|10\rangle),$
 and  $p_i\in(\frac{1}{3},1)$ for any $i\in\{1,2,3,4\}$.  In this network, the auxiliary local sources are all denoted by $S_v$, with $S_v=\frac{1}{3}|\psi^-\rangle\langle\psi^-|+
\frac{1}{6}I_2\otimes I_2$ for any $v\in\{5,6\}$. By setting $p_1=p_2=p_3=p_4=0.96$, we ensure that each source in $\Xi(4,4)$ is nonlocal. Using the SLSQP algorithm, we can always find a set of measurements such that $|I_{\rm new1}|^{\frac{1}{4}}+|J_{\rm new1}|^{\frac{1}{4}}=1.343$ holds, i.e., Ineq. (\ref{ex1}) is violated. This means that the correlations generated by the network are FQNN.

\begin{figure}
  \centering
\subfigure[]{\includegraphics[width=1.7in]{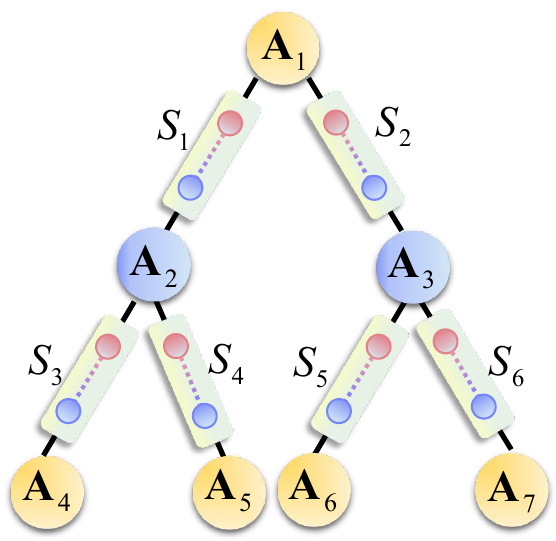}}
  \subfigure[]{\includegraphics[width=2.1
  in]{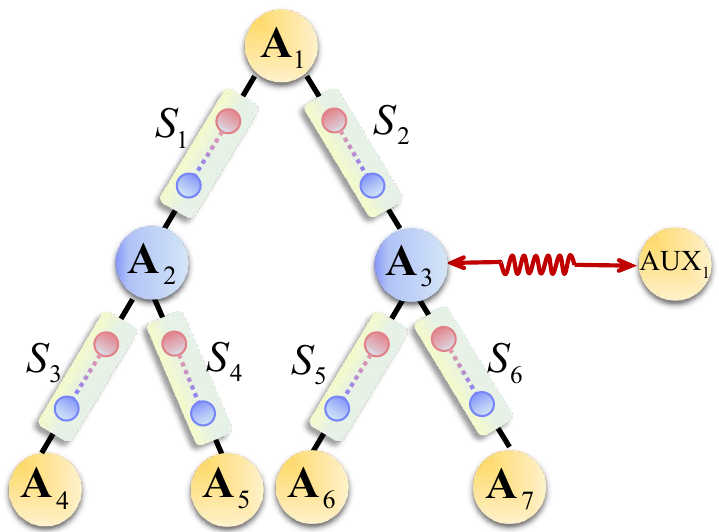}}
  \caption{  (a) A network $\Xi(7,6)$ with the maximum independence number 5. (b) The new network $\Xi^{\rm new}(7,6)$ formed after adding one auxiliary local source and one auxiliary party.  The additional party ${\rm AUX}_1$ is denoted as ${\bf A}_8$. 
   }\label{tree}
\end{figure}

\textit{Example 5}. For the 3-layer 2-forked tree-shaped network $\Xi(7,6)$ shown in Fig. \ref{tree}(a), it consists of 7 parties  ${\bf A}_{1}, {\bf A}_{2},\ldots, {\bf A}_{7}$ and 6 sources $S_1,S_2,\ldots,S_6$. The maximum independence number $h$ is 5, and the parties corresponding to the maximum independence number in the figure are represented by yellow circles. According to Steps 1-2 in this section, we obtain $\Gamma_2=\{1,4,5,6,7\}$,
$\bar\Gamma_2=\{2,3\}$ and $N(7,6)=0+6-5=1$, respectively.

Therefore, we need to add one additional local source $S_7$ and one additional party ${\bf A}_8$ in the network, forming a new network $\Xi^{\rm new}(7,6)$ as shown in Fig. \ref{tree}(b).  Clearly, the maximum independence number of $\Xi^{\rm new}(7,6)$ is 6. And $w=2$. According to Ineq. (\ref{gn}) in Step 3, one can obtain
\begin{align}\label{ex1}
|I_{\rm new 2}|^{\frac{1}{6}}+|J_{\rm new 2}|^{\frac{1}{6}}\leq
2^{\frac{5}{12}}
\end{align}
where $I_{\rm new 2}=\langle \Pi_{i\in\Gamma^{\rm new2}}A^+_{x_i}\Pi_{j\in\bar{\Gamma}^{\rm new2}}A_{x_j=0}\rangle$ and
$J_{\rm new 2}=\langle \Pi_{i\in\Gamma^{\rm new2}}A^-_{x_i}\Pi_{j\in\bar{\Gamma}^{\rm new2}}A_{x_j=1}\rangle$ with $\Gamma^{\rm new2}=\Gamma_2\cup\{8\}$ and $\bar{\Gamma}^{\rm new2}=\bar{\Gamma}_2=\{2,3\}$.
If each source in the network can be expressed in the form of $S_i=p_i|\psi^-\rangle\langle\psi^-|+\frac{1-p_i}{4}I_2\otimes I_2$ with
  $|\psi^-\rangle=\frac{1}{\sqrt{2}}(|01\rangle-|10\rangle)$, where $p_1=p_2=\cdots=p_6=0.95$ and $p_7=\frac{1}{3}$.

By employing the SLSQP algorithm, an optimal measurement setting can be found for each party, and the maximum value of $|I_{\rm new 2}|^{\frac{1}{6}}+|J_{\rm new 2}|^{\frac{1}{6}}$ is ultimately calculated to be $1.374$. Clearly, the inequality can be violated under these conditions, which demonstrates that our scheme is effective. This indeed verifies that the correlations generated by network $\Xi(7,6)$ are FQNN.

Examples 4 and 5 show that the value of the left-hand side of Ineq. (\ref{gn}) can be computed for any quantum network using the SLSQP algorithm. Thus, the observed violation of this inequality proves that all sources in the original network are nonlocal.

\section{Conclusion and  Discussion}

This paper proposes a scheme for determining whether the correlations are FQNN in a triangular network, primarily by considering the addition of auxiliary local sources and parties to the original network. This method only requires a single violation of Bell-like inequalities to reflect the FQNN of the network. Furthermore, the approach presented in this paper can be extended to chain networks, cyclic networks, and arbitrary quantum networks.

Compared to the decomposition method proposed in Ref. \cite{PhysRevA.110.022617}, the detection scheme provided in this paper significantly enhances the efficiency of detection, eliminating the need for each party to perform multiple repeated measurements. This reduction minimizes resource waste and measurement errors associated with complex measurement settings. The method provides a basis for experimentally verifying whether the correlations generated in different network structures are FQNN.

Additionally, the proposal we have put forward is anticipated to serve a critical function in the future, especially concerning the fault diagnosis within large-scale quantum networks. Quantum sources, due to their inherent fragility, are susceptible to degenerating from quantum states to classical sources, a process that could potentially trigger malfunctions throughout the entire network. Thus, detecting whether any quantum sources in the network have degenerated into classical ones has become an essential task. The method we propose for detecting correlations as FQNN  offers an efficient solution for fault detection in quantum networks.

\vskip 0.3cm {\bf Author Contributions}\quad  Professor Jinchuan Hou and Professor Kan He discussed the content of the paper and revised the paper. Professor Mingxing Luo provided suggestions for revision and modified the content of the paper. Shuyuan Yang was responsible for proving the theorems in the paper and completed the initial draft.

\vskip 0.3cm {\bf Data Availability} \quad  The data involved in the paper are available.

\vskip 0.3cm {\bf Conflict of Interest}\quad The authors declare no conflict of interest.



\bibliography{ref}
\bibliographystyle{unsrt}

\end{document}